\documentclass[a4paper,11pt,reqno]{amsart}

\usepackage{amssymb,amsmath,amsfonts} 
\usepackage{a4wide}
\usepackage{graphicx}
\usepackage{hyperref}
\usepackage{color}
\usepackage{amsthm}
\usepackage{tikz}
\usetikzlibrary{decorations,arrows}
\usetikzlibrary{decorations.pathmorphing}
\usepgflibrary{decorations.pathreplacing} 
\usepackage{pgf}
\usepackage{pgfkeys}
\usepackage{pgffor}
\usepackage{pgfpages}
\usepackage{enumerate}
\numberwithin{equation}{section}    
\theoremstyle{plain}
\newtheorem{Theorem}{Theorem}[section]
\newtheorem{Proposition}[Theorem]{Proposition}
\newtheorem{Corollary}[Theorem]{Corollary}
\newtheorem{Lemma}[Theorem]{Lemma}

\theoremstyle{definition}

\newtheorem{Definition}[Theorem]{Definition}

\theoremstyle{remark}
\newtheorem*{Remark}{Remark}

\renewcommand{\epsilon}{\varepsilon}
\renewcommand{\phi}{\varphi}
\newcommand{\one}{\mathbf{1}} 

\DeclareMathOperator{\NZ}{\mathbb{N}} 
\DeclareMathOperator{\RZ}{\mathbb{R}} 
\DeclareMathOperator{\GZ}{\mathbb{Z}} 
\DeclareMathOperator{\A}{\mathcal{A}} 
\DeclareMathOperator{\F}{\mathcal{F}} 

\DeclareMathOperator{\C}{\mathcal{C}} 

\DeclareMathOperator{\W}{\mathcal{W}} 
\DeclareMathOperator{\aw}{\mathit{a}_{\omega}} 
\DeclareMathOperator{\bw}{\mathit{b}_{\omega}} 
\DeclareMathOperator{\Hw}{\mathit{H}_{\omega}} 
\DeclareMathOperator{\wME}{\widetilde{\mathit{M}}^{\mathit{E}}} 
\DeclareMathOperator{\Mk}{\mathcal{M}_{\mathit{k}}} 

\newcommand{\hm}[1]{\textbf{*}\leavevmode{\marginpar{\tiny%
$\hbox to 0mm{\hspace*{-0.5mm}$\leftarrow$\hss}%
\vcenter{\vrule depth 0.1mm height 0.1mm width \the\marginparwidth}%
\hbox to 0mm{\hss$\rightarrow$\hspace*{-0.5mm}}$\\\relax\raggedright #1}}}
\title[Spectrum of Lebesgue Measure Zero for Jacobi Matrices]{Spectrum of Lebesgue Measure Zero for Jacobi Matrices of Quasicrystals}
\author[S.~Beckus, F.~Pogorzelski]{Siegfried Beckus, Felix Pogorzelski}
\address[S.B.]{Mathematisches Institut, Friedrich-Schiller-Universit\"at
Jena , 07743 Jena, Germany}
\urladdr{http://www.analysis-lenz.uni-jena.de/Team/Siegfried+Beckus.html}
\address[F.P.]{Mathematisches Institut, Friedrich-Schiller-Universit\"at
Jena , 07743 Jena, Germany}
\urladdr{http://www.analysis-lenz.uni-jena.de/Team/Felix+Pogorzelski.html}

\begin{document}

\begin{abstract}
We study one-dimensional random Jacobi operators corresponding to strictly ergodic dynamical systems. In this context, we characterize the spectrum of these operators by non-uniformity of the transfer matrices and the set where the Lyapunov exponent vanishes. Adapting this result to subshifts satisfying the so-called Boshernitzan condition, it turns out that the spectrum is supported on a Cantor set with Lebesgue measure zero. This generalizes earlier results for Schr\"odinger operators.
\end{abstract}
\maketitle

\section{Introduction}\label{section - Introduction}
Analyzing spectral properties of Schr\"odinger operators plays an important role in quantum mechanics, where one intends to study the long time behaviour of a particle in a space. The different sorts of potentials are a matter of particular interest in spectral theoretic research areas. Especially, one considers random Schr\"odinger operators which represent disordered solids. For instance, the periodic model and the Anderson model were intensively examined. In the first case, the potential is completely ordered and periodic. Hence, the corresponding model is appropriate to represent the molecular structure of crystals. In this situation, the spectrum is purely absolutely continuous. By contrast, the potential is absolutely random in the Anderson model, see \cite{An57}. In this context, some results for the spectrum of the corresponding Schr\"odinger operator are well-known. In detail, the works \cite{FMSS85,CKMF} show that it is purely discrete. Potentials which are aperiodic, i.e.\@ ordered but not periodic, can thematically be classified between these two models. The examination of such potentials has soared up after the year 1982 when Dan Shechtman discovered quasicrystals, see \cite{SBGC}. First considerations about one-dimensional Schr\"odinger operators with quasiperiodic potentials can for example be found in \cite{OstlundPanditRandSchellnhuberSiggiai,KohmotoKadamoffTang}.

\medskip

From the mathematical point of view, many spectral questions concerning Schr\"odinger operators with aperiodic potential arise. One issue is to show that the spectrum is a Cantor set of Lebesgue measure zero. If this condition is satisfied we will write (C). In the last decades, two classes of models have attracted special attention in the discrete, one-dimensional case. 
On the one hand, the class of dynamical systems induced by substitutions has widely been studied. First results about the Fibonacci substitution can be found in \cite{Suto87,Suto89}. In these works, condition (C) is shown under certain conditions. By using trace maps, the absence of point spectrum and (C) is shown for a larger class of primitive substitutions with some reasonable requirements in the works \cite{BellissardBovierGhez,BovierGhez}. In \cite{Damanik98}, the almost sure absence of eigenvalues was shown for primitive substitutions with the property that the potentials have a local four block structure. On the other hand, the consideration of potentials induced by circle maps has drawn particular attention to itself. Precisely, it is shown in the paper \cite{BellissardIochumScopullaTestard} that for irrational numbers, the spectrum is equal to the set where the Lyapunov exponent vanishes. 
Using this and adapting the Kotani result (cf.~\cite{Kotani}) the absence of absolutely continuous spectrum follows. Further results on upper bounds on the growth of solutions, as well as on the fact that the point spectrum is empty for all Sturmian potentials can be found in \cite{DamanikLenzI,DamanikLenzII,DamanikKillipLenz}.

\medskip

By using general techniques the work \cite{Lenz1} proves for a large class of substitutions and Sturmian systems that the corresponding family of Schr\"odinger operators satisfies (C). In detail, the paper contains a characterization of the spectrum by the Lyapunov exponent and non-uniform transfer matrices. 
Then (C) holds for a Schr\"odinger operator induced by a subshift, if the transfer matrices are all uniform. It turns out that the so called Boshernitzan condition, first introduced by Boshernitzan (\cite{Boshernitzan1}), for subshifts is suitable to show the uniformity of the transfer matrices, see \cite{DamanikLenz}. Indeed, a large class of models fulfills this condition, such as subshifts satisfying a positive weight condition, all Sturmian subshifts, almost all interval exchange transformations, almost all circle maps  and almost all Arnoux-Rauzy subshifts, see \cite{DamanikLenz,DamanikLenz2}.

\medskip

The aim of this article is to extend the results of \cite{Lenz1} to random Jacobi operators arising from a strictly ergodic topological dynamical system $(\Omega,T)$. We consider two continuous functions $p:\Omega\rightarrow\RZ\setminus\{0\}$ and $q:\Omega\rightarrow\RZ$ and its corresponding Jacobi operator on $\ell^2(\GZ)$
\begin{gather*}
(\Hw u)(n) := p(T^n\omega) \cdot u(n-1) + p(T^{n+1}\omega) \cdot u(n+1) + q(T^n\omega) \cdot u(n),\quad n\in\GZ.
\end{gather*}                                                                                                  
This is the discrete version of a Schr\"odinger operator with weighted Laplacian. It is well-known that there exists a closed subset $\Sigma\subseteq\RZ$ such that for all $\omega\in\Omega$ the equality $\sigma(\Hw)=\Sigma$ holds, see e.g. \cite{Lenz3}. Our purpose is to characterize the spectrum $\Sigma$ by the non-uniformity of the transfer matrices and the set where the Lyapunov exponent $\gamma:\RZ\to[0,\infty)$ vanishes. In particular, we will verify in our setting that
\begin{gather*}
\Sigma=\{E\in\RZ\;|\; \gamma(E)=0\}\;\bigsqcup\;\{E\in\RZ\;|\; M^E \text{ is not uniform}\},\quad\quad\quad\quad(\clubsuit)
\end{gather*}
where $M^E$ is the transfer matrix corresponding to the energy $E\in\RZ$. Applying this result to subshifts and using \cite{Remling} we will get the following statement: Let $(\Omega,T)$ be an aperiodic, strictly ergodic dynamical system. If the transfer matrices $M^E$ are uniform for all energies $E\in\RZ$ and if $p$ and $q$ take only finitely many values, it follows that $\Sigma$ fulfills (C), see Theorem~\ref{Theorem - aperiodic subshifts has spectrum of Lebesgue measure zero}. Since the images of $p$ and $q$ are finite sets, we can apply the result of \cite{DamanikLenz} stating that the Boshernitzan condition of a subshift is sufficient for the transfer matrices $M^E$ beeing uniform for all energies $E\in\RZ$. 

\medskip

Results of this kind have only been proven for some special cases so far. In detail, in \cite{Yessen} 
it is shown among other sophisticated results that for Fibonacci sequences with a coupling constant the spectrum is a Cantor set of Lebesgue measure zero. The works \cite{DamanikGorodetski} 
and \cite{Dahl} prove (C) for Jacobi operators associated with the Fibonacci sequence with vanishing potential $q$ and positive values of the alphabet. Further elaborations can for instance be found in \cite{JanasNabokoStolz,Marin}.

\medskip

The paper is organized as follows. In Section~\ref{section - Generalities}, we introduce the relevant objects for our model, i.e.\@ the notions of strictly ergodic dynamical systems and of cocycles for continuous, matrix-valued functions. Next, we examine these objects in more detail in Section~\ref{section - Key Results}. This includes growth behaviour and continuity properties. Section~\ref{section - Jacobi matrices} is devoted to random Jacobi operators induced by a strictly ergodic dynamical system. In this context, we prove important connections between the Lyapunov exponent of the corresponding transfer matrices and the spectral properties of the operator. These considerations lead to our main result in Section \ref{section - The main results}, where we give a complete description of the spectrum in terms of the Lyapunov exponent and the uniformity property of the transfer matrices. Precisely, we show the Equality $(\clubsuit)$, cf.\@ Theorem \ref{Theorem - the spectrum is equal to the union}. Finally, we apply our results to subshifts in Section~\ref{section - Subshifts}. In fact, we show in Theorem \ref{Theorem - aperiodic subshifts has spectrum of Lebesgue measure zero} that for a very large class of operator families, uniformity of all transfer matrices implies (C).

\section{Generalities}\label{section - Generalities}

We start by defining the relevant objects for our work. To do so, we introduce the notion of cocycles induced by some measure preserving, strictly ergodic transformation on a topological probability space. Further, we apply a version of Kingman's subadditive ergodic theorem in order to get some almost-sure approximation results for the growth rate of the underlying matrix norms.
  
\medskip

Let $(\Omega,T)$ be a dynamical system, where $\Omega$ is a compact metric space and $T:\Omega\rightarrow\Omega$ is a homeomorphism. Let $\mu$ be a probability measure on the Borel $\sigma$-algebra on $\Omega$. The measure $\mu$ is called \textit{invariant}, if for all $A\in\F$
\begin{gather*}
\mu(T(A))=\mu(A).
\end{gather*}
The dynamical system is called \textit{ergodic}, if each measurable $A$ with $A=T^{-1}(A)$ has measure one or zero. A dynamical system is called \textit{uniquely ergodic} if there exists only one invariant ergodic probability measure on $\F$. Further, it is called \textit{minimal}, if every orbit is dense in $\Omega$. If $(\Omega,T)$ is both uniquely ergodic and minimal, it is called \textit{strictly ergodic}.

\medskip

Consider the general linear group $GL(2,\RZ)$ of 2x2 matrices with real values and nonzero determinant and the special linear group $SL(2,\RZ)$ as the subgroup of those matrices with determinant one. The topology on these groups is defined by the operator norm $\|\cdot\|$. For a continuous map $M:\Omega\rightarrow GL(2,\RZ)$ we define the cocycle $M(n,\omega)$ for $\omega\in\Omega$ and $n\in\GZ$ by
\begin{align*}
M(n,\omega) :=\begin{cases}
M(T^{n-1}\omega)\cdot\ldots\cdot M(\omega)\quad &: n>0\\
Id\quad &: n=0\\
M^{-1}(T^{n}\omega)\cdot\ldots\cdot M^{-1}(T^{-1}\omega)\quad &: n<0.\\
\end{cases}
\end{align*}
Note that the equality
\begin{gather*}
M(m,T^n\omega)\cdot M(n,\omega) = M(m+n,\omega)
\end{gather*}
holds for each $m,n\in\GZ$ and $\omega\in\Omega$. The following proposition states a well-known version of Kingman's subadditive ergodic theorem, see e.g. \cite{KatznelsonWeiss}.

\begin{Proposition} \label{Proposition - Kingmans subadditive ergodic theorem}
Let $(\Omega,T)$ be uniquely ergodic with invariant probability measure $\mu$ and $M:\Omega\rightarrow GL(2,\RZ)$ be continuous. Then for
\begin{gather*}
\Lambda(M) := \inf\limits_{n\in\NZ}\frac{1}{n}\cdot \int\limits_{\Omega} log\|M(n,\omega)\|\;d\mu(\omega)
\end{gather*}
the equality
\begin{gather*}
\Lambda(M) = \lim_{n\rightarrow\infty} \frac{1}{n} \cdot log\|M(n,\omega)\|
\end{gather*}
holds for $\mu$-a.e. $\omega\in\Omega$.
\end{Proposition}

Following in \cite{Furman97}, we use the following definition. It is motivated by the fact that unique ergodicity of $(\Omega,T)$ is equivalent to the uniform convergence of the ergodic averages in the case of continuous functions.
\begin{Definition}
Let $(\Omega,T)$ be strictly ergodic. A continuous map $M:\Omega\rightarrow GL(2,\RZ)$ is called \textit{uniform}, if the limit 
\begin{gather*}
\Lambda(M)=\lim_{n\rightarrow\infty}\frac{1}{n}\cdot log\|M(n,w)\|
\end{gather*}
exists for all $\omega\in\Omega$ and converges uniformly in $\omega\in\Omega$.
\end{Definition}

Note that in the paper of \cite{Furman97}, it is not required that the dynamical system $(\Omega,T)$ is minimal. However, it is convenient for our setting to assume minimality. For minimal topological dynamical systems, uniform existence of the limit implies uniform convergence as shown by Weiss, cf.\@\cite{Lenz2} as well.

\section{Key Results}\label{section - Key Results}

In this section, we provide general facts for strictly ergodic cocycles $M: \Omega \rightarrow SL(2, \mathbb{R})$. They will be exploited in the proofs of the spectral theoretic statements for Jacobi operators in the following Sections \ref{section - Jacobi matrices} and \ref{section - The main results}.      

\medskip

The next assertion can be found for example in \cite{Lenz1}, Lemma 3.2.
\begin{Lemma}\label{Lemma - exponential growth}
Let $(\Omega,T)$ be a dynamical system which is strictly ergodic with invariant probability measure $\mu$. Consider a uniform $M:\Omega\rightarrow SL(2,\RZ)$ with $\Lambda(M)>0$. Then for each $u\in\RZ^2\setminus\{0\}$ and $\omega\in\Omega$ there are constants $D,\kappa>0$ such that
\begin{gather*}
\|M(n,\omega)u\|\geq D\cdot e^{\kappa\cdot|n|}
\end{gather*}
for all $n\geq 0$ or $n\leq 0$.
\end{Lemma}

The following notion is inspired by \cite{Furman97} and \cite{Lenz2}. Let $\mathcal{U}(\Omega)$ be the set of uniform, continuous maps $M:\Omega\to SL(2,\RZ)$. Further, the set $\mathcal{U}(\Omega)_{+}$ are the elements $M\in\mathcal{U}(\Omega)$ where $\Lambda(M)>0$. A metric on the complete metric space of continuous maps defined on $\Omega$ with values in $SL(2,\RZ)$ is given by
\begin{gather*}
d(A,B):= \sup_{\omega\in\Omega}\|A(\omega)-B(\omega)\|,
\end{gather*}
see \cite{Furman97}. Denote this complete metric space by $\C(\Omega,SL(2,\RZ))$.

\medskip

The proof of Theorem \ref{Theorem - Offenheit von U+(Omega)} can be found in \cite{Furman97}, cf.\@\cite{Lenz2} as well.

\begin{Theorem}\label{Theorem - Offenheit von U+(Omega)}
Let $(\Omega,T)$ be strictly ergodic. Then the set $\mathcal{U}(\Omega)_{+}$ is open in the space $\C(\Omega,SL(2,\RZ))$ and the map $\Lambda:\mathcal{U}(\Omega)\to\RZ$ is continuous.
\end{Theorem}

The next assertion is an adaption of a result of \cite{Furman97}, see \cite{Lenz1} as well.

\begin{Lemma}\label{Lemma - Convergence of the Lambda}
Let $(\Omega,T)$ be strictly ergodic. Consider a uniform $M:\Omega\rightarrow SL(2,\RZ)$ and a sequence of continuous maps $M_n:\Omega\rightarrow SL(2,\RZ)$ where $d(M_n,M)$ tends to zero. Then
\begin{gather*}
\Lambda(M_n)\overset{n\rightarrow\infty}{\longrightarrow}\Lambda(M).
\end{gather*}
\end{Lemma}

\begin{proof}
As a consequence of Theorem \ref{Theorem - Offenheit von U+(Omega)} the convergence of $\Lambda(M_n)$ to $\Lambda(M)$ follows if $d(M_n,M)$ and $d(M^{-1}_n,M^{-1})$ converge to zero. If $\lim\limits_{n\to\infty}d(M_n,M)=0$ a short computation leads to $\lim\limits_{n\to\infty} d(M^{-1}_n,M^{-1}) = 0$ .
\end{proof}

The next statements provide useful tools to show the uniformity in some situations which are convenient for our purpose. Lemma \ref{Lemma - Boundeness} provides an upper bound for the logarithmic growth of the norm of a continuous map $M:\Omega\rightarrow GL(2,\RZ)$, see \cite{Furman97}, Corollary 2.

\begin{Lemma} \label{Lemma - Boundeness}
Let $(\Omega,T)$ be uniquely ergodic and consider a continuous $M:\Omega\rightarrow GL(2,\RZ)$. Then
\begin{gather*}
\limsup_{n\rightarrow\infty}\frac{1}{n} \cdot log\|M(n,\omega)\| \leq \Lambda(M) 
\end{gather*}
uniformly on $\Omega$.
\end{Lemma}

\begin{Lemma}\label{Lemma - equivalence vanishing of Lambda and uniformity}
Consider a uniquely ergodic dynamical system $(\Omega,T)$. Let $M:\Omega\to GL(2,\RZ)$ and $\widetilde{M}:\Omega\rightarrow GL(2,\RZ)$ be continuous. If there is a constant $K\geq 1$ independent of $\omega\in\Omega$ and $n\in\GZ$ such that 
\begin{gather*}
\|M(n,\omega)\|\leq K\cdot \|\widetilde{M}(n,\omega)\| \text{ and } \|\widetilde{M}(n,\omega)\|\leq K\cdot \|M(n,\omega)\|,
\end{gather*}
then the following two statements hold.
\begin{itemize}
\item[(i)]  The equality $\Lambda(M)=\Lambda(\widetilde{M})$ holds.
\item[(ii)] The map $M$ is uniform, if and only if $\widetilde{M}$ is uniform as well.
\end{itemize}
\end{Lemma}

\begin{proof}
This proof is straight forward by a short computation.
\end{proof}

\begin{Lemma}\label{Lemma - stetig diagonalisierbar}
Let $M:\Omega\rightarrow GL(2,\RZ)$ be continuous such that
\begin{gather*}
M(\omega) = C^{-1}(T\omega) \cdot \left(\begin{matrix}
f_1(\omega) & 0\\
0 & f_2(\omega)
\end{matrix}\right)\cdot C(\omega)
\end{gather*}
where $|f_1|,|f_2|\in\C(\Omega,\RZ)$ and $C:\Omega\rightarrow GL(2,\RZ)$ is such that $\|C\|,\|C^{-1}\|:\Omega\to\RZ$ are continuous. Then the function $M$ is uniform.
\end{Lemma}

\begin{proof}
First note that for $\omega\in\Omega$ the equality 
\begin{gather*}
M(n,\omega)= C^{-1}(T^n\omega)\cdot \underbrace{\begin{pmatrix}
\prod\limits_{j=0}^n f_1(T^j\omega) & 0\\
0 & \prod\limits_{j=0}^n f_2(T^j\omega)
\end{pmatrix}}_{=:A(n,\omega)} C(\omega), \quad\quad n\in\NZ
\end{gather*}
holds. Since $\|C\|,\|C^{-1}\|:\Omega\rightarrow \RZ$ are continuous and $\Omega$ is compact, we immediately get that there are constants $C_1,C_2>0$ such that $C_1\cdot \|A(n,\omega)\|\leq \|M(n,\omega)\|\leq C_2\cdot \|A(n,\omega)\|$ for all $n\in\NZ$ and every $\omega\in\Omega$.

\medskip

Moreover, we know that all norms defined on the linear space of matrices are equivalent and so, there are constants $D_1,D_2>0$ such that
\begin{gather*}
D_1\cdot \max\left\{ \prod\limits_{j=0}^n |f_1(T^j\omega)|, \prod\limits_{j=0}^n |f_2(T^j\omega)| \right\}\leq \|M(n,\omega)\|\leq D_2 \cdot\max\left\{ \prod\limits_{j=0}^n |f_1(T^j\omega)|, \prod\limits_{j=0}^n |f_2(T^j\omega)| \right\}
\end{gather*}
Since $M$ is invertible it follows that $f_1$ and $f_2$ never vanish and so, the functions $\log\circ|f_1|$ and $\log\circ|f_1|$ are well-defined and continuous by our requirements. Using unique ergodicity, it follows from standard arguments in ergodic theory (see e.g.\@ \cite{EisnerNagel}, Corollary 9.9) that $\frac{1}{n}\sum\limits_{j=0}^n \log\circ|f_1|\circ T^j$ and $\frac{1}{n}\sum\limits_{j=0}^n \log\circ|f_2|\circ T^j$ converge uniformly on $\Omega$ to some constant. Hence, $\frac{1}{n}\log \|M(n,\omega)\|$ converge uniformly on $\Omega$ to a constant. 
\end{proof}

\section{Jacobi matrices}\label{section - Jacobi matrices}

We present the notion of a uniquely ergodic family $\{H_{\omega}\}_{\omega \in \Omega}$ of Jacobi operators acting on $\ell^2(\mathbb{Z})$. For each $E \in \mathbb{R}$, we define the transfer matrices $M^E$ corresponding to the equation $(H_{\omega}-E)\,u = 0$ as cocycle functions on $\Omega$. Further, we verify the existence of the Lyapunov exponent $\gamma(E)$ containing information on the growth rate of the matrix norms, cf.~Lemma \ref{Lemma - equivalence of M und tilde(M)}. By drawing connections between the Lypunov exponent and the spectrum of the operators, we prove major preparations for the main results of this work (cf.~Section \ref{section - The main results}) in the Proposition \ref{Proposition - E nicht im Spektrum}, Lemmas \ref{Lemma - M uniform impliziert} and \ref{Lemma - E with gamma(E)neq 0 implies uniformity}. 

\medskip

We consider two continuous maps
$p:\Omega\rightarrow\RZ\setminus\{0\}$ and $q:\Omega\rightarrow\RZ$ and for $\omega\in\Omega$ its corresponding Jacobi operator $H_{\omega}:\ell^2(\GZ)\to\ell^2(\GZ)$ defined by
\begin{gather*}
(\Hw u)(n) := p(T^{n+1}\omega) \cdot u(n+1) + q(T^n\omega) \cdot u(n) + p(T^n\omega) \cdot u(n-1).
\end{gather*}                                                                                                  
Denote $p(T^n\omega)$ by $\aw(n)$ and $q(T^n\omega)$ by $\bw(n)$. Since $\Omega$ is compact and $p$ and $q$ are continuous there is a constant $K\geq 1$ such that
\begin{gather*}
\frac{1}{K}\leq |a_{\omega}(\cdot)|\leq K
\end{gather*}
and
\begin{gather*}
0\leq |b_{\omega}(\cdot)|\leq K
\end{gather*}
for each $\omega\in\Omega$. Using this boundeness it follows that the norm of $\Hw$ is bounded, because
\begin{gather*}
\|\Hw\|\leq 2\cdot \|\aw\|_{\infty} + \|\bw\|_{\infty}\leq 3\cdot K.
\end{gather*}

\medskip

For $\omega\in\Omega$ and $E\in\RZ$, we are interested in general solutions of the difference equation
\begin{gather*}
\aw(n+1) \cdot u(n+1) + \bw(n) \cdot u(n) + \aw(n) \cdot u(n-1) - E \cdot u(n) = 0\qquad\qquad(\spadesuit)
\end{gather*}
for $n\in\GZ$. We define the so called \textit{transfer matrix} by
\begin{gather*}
M^E(\omega):=\left(\begin{matrix}
\frac{E-\bw(1)}{\aw(2)} & -\frac{\aw(1)}{\aw(2)}\\
1 & 0
\end{matrix}\right).
\end{gather*}
Similarly to the Schr\"odinger case it follows that $(\spadesuit)$ holds, if and only if
\begin{gather*}
\left(\begin{matrix}
u(n+1)\\
u(n)
\end{matrix}\right)= M^E(n,\omega)\cdot \left(\begin{matrix}
u(1)\\
u(0)
\end{matrix}\right)
\end{gather*}
for all $n\in\GZ$. Unlike to the classical case of Schr\"odinger operators the determinant of $M^E(\omega)$ is not necessarily equal to one. Thus, we introduce the following matrix
\begin{gather*}
\wME(\omega) :=\left(\begin{matrix}
\frac{E-\bw(1)}{\aw(2)} & -\frac{1}{\aw(2)}\\
\aw(2) & 0
\end{matrix}\right)
\end{gather*}
with determinant equal to one. Then the equation $(\spadesuit)$ holds, if and only if
\begin{gather*}
\left(\begin{matrix}
u(n+1)\\
\aw(n+1)\cdot u(n)
\end{matrix}\right)= \widetilde{M}^E(n,\omega)\cdot \left(\begin{matrix}
u(1)\\
\aw(1) \cdot u(0)
\end{matrix}\right).
\end{gather*}

Note that the maps $M^E$ and $\wME$ are continuous by definition.

\begin{Lemma}\label{Lemma - equivalence of M und tilde(M)}
Let $(\Omega,T)$ be strictly ergodic and consider the maps\- $M^E:\Omega\to GL(2,\RZ)$\- and $\wME:\Omega\to SL(2,\RZ)$ defined as above. Then for an $\omega\in\Omega$ the limit 
\begin{gather*}
\lim\limits_{n\to\infty}\frac{1}{n}\log\|M^E(n,\omega)\|
\end{gather*}
exists if and only if the limit
\begin{gather*}
\lim\limits_{n\to\infty}\frac{1}{n}\log\|\wME(n,\omega)\|
\end{gather*}
exists and in these cases, they are equal. Moreover, $M^E$ is uniform if and only if $\wME$ is uniform.
\end{Lemma}

\begin{proof}
Define the continuous map $C:\Omega\to GL(2,\RZ)$ by
\begin{gather*}
C(\omega):=\left(\begin{matrix}
1 & 0\\
0 & \aw(1)
\end{matrix}\right).
\end{gather*}
Then we get the equation $M^E(\omega)=C^{-1}(\omega)\wME(\omega) C(\omega)$. This yields to the equality $M^E(n,\omega) = C^{-1}(T^{n}\omega)\wME(n,\omega) C(\omega)$. By Lemma \ref{Lemma - equivalence vanishing of Lambda and uniformity} our statements follows.
\end{proof}

Accoring to Proposition \ref{Proposition - Kingmans subadditive ergodic theorem} we can define the \textit{Lyapunov exponent} for the energy $E\in\RZ$ by $\gamma(E) := \Lambda(M^E)$.

\begin{Lemma}\label{Lemma - positivity of the lyapunov exponents}
The Lyapunov exponent $\gamma(E)$ is greater or equal than zero for all $E\in\RZ$.
\end{Lemma}

\begin{proof}
Since $\det(\wME(T^n\omega))=1$ for each $n\in\NZ$ it follows that $\det(\wME(n,\omega))=1$ for all $n\in\NZ$. Thus, the norm $\|\wME(n,\omega)\|$ is greater or equal than one and so 
\begin{gather*}
\gamma(E)\overset{\text{L. \ref{Lemma - equivalence of M und tilde(M)}}}{=}\Lambda(\wME)\overset{\text{L. \ref{Lemma - Boundeness}}}{\geq}\limsup_{n\rightarrow\infty}\frac{1}{n}\cdot\log(\underbrace{\|\wME(n,\omega)\|}_{\geq 1})\geq 0.
\end{gather*}
\end{proof}

The following well-known Proposition \ref{Proposition - Constancy of the spectrum} states that the spectrum of the Jacobi operators with respect\- to the elements of $\Omega$ does not change, if the dynamical system $(\Omega,T)$ is minimal. Hence, one also can talk about the spectrum of the whole family of Jacobi operators, see e.g. \cite{Lenz3}.

\begin{Proposition} \label{Proposition - Constancy of the spectrum}
Let $(\Omega,T)$ be a minimal dynamical system. Then there exists a set $\Sigma\subseteq\RZ$ such that $\sigma(\Hw)=\Sigma$ for every $\omega\in\Omega$.
\end{Proposition}

As a direct consequence of Lemma \ref{Lemma - exponential growth} and some general result, Proposition \ref{Proposition - E nicht im Spektrum} states a sufficient condition for $E\in\RZ$ not to be contained in the spectrum of a family of Jacobi operators.

\begin{Proposition}\label{Proposition - E nicht im Spektrum}
Let $(\Omega,T)$ be strictly ergodic and $p:\Omega\to\RZ\setminus\{0\}$ and $q:\Omega\to\RZ$ be continuous maps with corresponding family of Jacobi operators $(H_{\omega})_{\omega\in\Omega}$. If for $E\in\RZ$ we have $\Lambda(M^{E})>0$ and $M^{E}$ is uniform, then $E$ does not belong to the spectrum $\Sigma$ of the family of Jacobi operators $(H_{\omega})_{\omega\in\Omega}$.
\end{Proposition}

\begin{proof}
By Lemma \ref{Lemma - equivalence of M und tilde(M)} $\widetilde{M}^{E}$ is uniform and $\Lambda(\widetilde{M}^{E})>0$. Recall the definition of the metric $d$ on $\C(\Omega,SL(2,\RZ)$ defined in the beginning of Section \ref{section - Key Results}. Note that with respect to this metric, for each $\varepsilon>0$ there exists an interval $I(\varepsilon)\ni E$ such that for all $F\in I(\varepsilon)$ we have $d(\widetilde{M}^F,\wME)<\varepsilon$. According to Theorem \ref{Theorem - Offenheit von U+(Omega)} the set \begin{gather*}
\mathcal{U}(\Omega)_+:=\{M:\Omega\to SL(2,\RZ)\;|\; M \text{ continuous, uniform and } \Lambda(M)>0\}
\end{gather*}
is open. Thus, there exists an open interval $I$ containing $E$ such that for all $F\in I$ we have $M^F\in\mathcal{U}(\Omega)_+$.

\medskip

Choose one $\omega\in\Omega$ and assume the contrary i.e. $E\in\Sigma$. Thus, there exists spectrum of $\Hw$ in $I$. Consequently, the spectral measure of $\Hw$ gives actually weight to $I$. Hence, there must be a solution $(u(n))_{n\in\GZ}$ 
 of the difference equation $(\spadesuit)$ for one $F\in I$ which is polynomially bounded and not zero, see \cite{CarmonaLacroix}, Theorem II.4.5. Since $\widetilde{M}^{F}\in\mathcal{U}(\Omega)_+$ and $(u(n))_{n\in\GZ}$ solves $(\spadesuit)$ there are constants $K\geq 1$ and $\kappa,D>0$ such that 
\begin{gather*}
K\cdot\left\| \left(\begin{matrix}
u(n+1)\\
u(n)
\end{matrix}\right) \right\| \geq\left\| \begin{pmatrix}
u(n+1)\\
\aw(n+1)\cdot u(n)
\end{pmatrix} \right\|= \left\|\widetilde{M}^{F}(n,\omega)\begin{pmatrix}
u(1)\\
\aw(1)\cdot u(0)
\end{pmatrix}\right\|\geq D\cdot e^{\kappa\cdot |n|}
\end{gather*}
for $n\geq 0$ or $n\leq 0$, see Lemma \ref{Lemma - exponential growth}. This contradicts the polynomial boundedness of $(u(n))_{n\in\GZ}$.
\end{proof}

The next step will be to prove a general result that the set where the Lyapunov exponent vanishes is contained in the spectrum under certain conditions to the dynamical system. In order to do so, we use the notion of a subexponentially increasing sequence. In detail, $(u(n))_{n\in\GZ}$ is called \textit{subexponentially increasing}, if 
\begin{gather*}
\limsup\limits_{|n|\to\infty} \frac{1}{|n|}\log|u(n)|\leq 0.
\end{gather*}
By some elementary arguments, one can check that this is equivalent to the fact that
\begin{gather*}
\limsup\limits_{|n|\to\infty} \frac{1}{|n|}\log\left(\sum\limits_{j=-|n|}^{|n|}|u(j)|^2\right)^{\frac{1}{2}}\leq 0.
\end{gather*}

Define for $(u(n))_{n\in\GZ}$ a new sequence
\begin{gather*}
u_l(k):=\one_{[-l,l]}(k)\cdot u(k),\quad\quad\quad k\in\GZ
\end{gather*}
where $\one_{[-l,l]}(k)$ is equal to one if $k\in [-l,l]$ and zero else. 

\begin{Lemma}\label{Lemma - Property of sequences which increas subexponential}
Let $(u(n))_{n\in\GZ}\subseteq\RZ$ be a subexponentially increasing sequence. Then for all $\delta>0$ there is some $n(\delta)\in\NZ$ such that for each $l\in\NZ$ with $l\geq n(\delta)$ we have
\begin{gather*}
\|u_{l+1}\|^2\leq e^{\delta}\cdot \|u_{l-1}\|^2
\end{gather*}
\end{Lemma}

\begin{proof}
Assume the contrary, which means that for every $\delta>0$ and $l(\delta)\in\NZ$ there is an $l\geq l(\delta)$ such that $\|u_{l+1}\|^2 > e^{\delta} \|u_{l-1}\|^2$. Without loss of generality, consider the subsequence $(l_k)_{k\in\NZ}$ such that $l_k$ is even for each $k\in\NZ$ and such that the latter inequality holds. Then $\|u_{l_k+1}\| > e^{\delta\cdot\frac{l_k}{2}}\cdot\|u_1\|$
and so,
\begin{align*}
\frac{1}{l_k+1}\cdot\log\|u_{l_k+1}\| > 
 \frac{\delta}{2}\cdot\underbrace{\frac{l_k}{l_k+1}}_{\geq \frac{1}{2}} + \frac{1}{l_k+1}\cdot\log\|u_1\|
\geq  \frac{\delta}{4} + \frac{1}{l_k+1}\log\|u_1\|.
\end{align*}
Hence,
\begin{gather*}
\limsup_{k\rightarrow\infty}\frac{1}{l_k+1}\log\|u_{l_k+1}\|
\geq\liminf_{k\rightarrow\infty}\frac{1}{l_k+1}\log\|u_{l_k+1}\|> \frac{\delta}{4}>0
\end{gather*}
contradicting the subexponential growth.
\end{proof}

The following well-known statement can for example be found in \cite{CarmonaLacroix}, Proposition V.4.1. We give the proof for the sake of completeness.

\begin{Lemma}\label{Lemma - Gamma is a subset of the spectrum}
Let $(\Omega,T)$ be strictly ergodic. Then
\begin{gather*}
\Gamma:=\{E\in\RZ\;|\; \gamma(E)=0\} \subseteq\Sigma.
\end{gather*}
\end{Lemma}

\begin{proof}
Consider some $E\in\Gamma$ and choose one $\underline{u}\in\RZ^2$ as initial condition with $\|\underline{u}\|=1$. Let $\left(u(n)\right)_{n\in\GZ}$ be a solution of the difference equation $(\spadesuit)$ for some $\omega\in\Omega$ with $\begin{pmatrix}
u(1)\\
u(0)
\end{pmatrix}:=\underline{u}$. Then Lemma \ref{Lemma - equivalence vanishing of Lambda and uniformity} and a short computation lead to $\Lambda(M^E)=\Lambda\left(\left(M^E\right)^{-1}\right)$. Consequently, the inequality
\begin{gather*}
\limsup_{|n|\rightarrow\infty}\frac{1}{|n|}\log\left\|\left(\begin{matrix}
u(n)\\
u(n+1)
\end{matrix}\right)\right\| \leq 0
\end{gather*}
follows by Lemma \ref{Lemma - Boundeness} meaning that $(u(n))_{n\in\GZ}$ is subexponentially increasing.

\medskip

Recapitulate the notion of $u_l(n):=\one_{[-l,l]}(n)\cdot u(n)$ ($n\in\GZ$) for $l\in\NZ$. Then by using the subexponential growth of $(u(n))_{n\in\GZ}$ there is for all $\delta>0$ some $l(\delta)\in\NZ$ such that for each $l\in\NZ$ with $l\geq l(\delta)$ it follows

\begin{align*}
\|(\Hw-E)u_l\|^2  \underset{C:=\|\Hw\|+|E|}{\leq} &C\cdot\left( \underbrace{\|u_{l+1}\|^2}_{\leq e^{\delta}\|u_{l-1}\|^2, \text{L. \ref{Lemma - Property of sequences which increas subexponential}}} - \;\;\|u_{l-1}\|^2 \right)\\
\overset{\quad\quad\quad\quad\quad}{\leq} &C (e^{\delta}-1)\cdot \|u_{l}\|^2.
\end{align*}

\medskip

Since the expression $\left(e^{\delta}-1\right)$ converges to zero as $\delta$ tends to zero we can choose a diagonal subsequence $(l_k)_{k\in\NZ}$ such that
\begin{gather*}
\left\|(\Hw-E)\frac{u_{l_k}}{\|u_{l_k}\|}\right\|\overset{k\rightarrow\infty}{\longrightarrow} 0.
\end{gather*}
Thus, we have constructed a Weyl sequence for $E\in\Gamma$ with respect to the operator $\Hw$. Hence, by general results, it follows that $E$ is an element of the spectrum $\Sigma$.
\end{proof}

Proposition \ref{Proposition - E nicht im Spektrum} and Lemma \ref{Lemma - Gamma is a subset of the spectrum} yield to the following statement.

\begin{Lemma}\label{Lemma - M uniform impliziert}
Let $(\Omega,T)$ be strictly ergodic. If $M^E$ is uniform for every $E\in\RZ$, then $\Sigma=\Gamma$ and the Lyapunov exponent $\gamma:\RZ\rightarrow[0,\infty)$ is continuous.
\end{Lemma}

\begin{proof}
The equation $\Sigma=\Gamma$ is a direct consequence of Proposition \ref{Proposition - E nicht im Spektrum} and Lemma \ref{Lemma - Gamma is a subset of the spectrum}. A short calculation using Lemma \ref{Lemma - Convergence of the Lambda} shows that $\Lambda(\wME)$ is continuous. By Lemma \ref{Lemma - equivalence of M und tilde(M)}, this holds also true for $\gamma$. 
\end{proof}

\begin{Lemma}\label{Lemma - E with gamma(E)=0 implies uniformity}
Let $(\Omega,T)$ be strictly ergodic. For $E\in\RZ$ with $\gamma(E)=0$ it follows that $M^E$ is uniform.
\end{Lemma}

\begin{proof}
This follows immediately from Lemma \ref{Lemma - Boundeness}, Lemma \ref{Lemma - equivalence of M und tilde(M)} and Lemma \ref{Lemma - positivity of the lyapunov exponents}.
\end{proof}

The following statement, well-known under the name Combes/Thomas argument, can be proven along the lines of \cite{Kirsch} (Theorem 11.2) by adjusting constants. For the convenience of the reader, we give a sketch of the proof.

\begin{Proposition}\label{Proposition - Combes Thomas}
Let $(\Omega,T)$ be a dynamical system and let $\left(\Hw \right)_{\omega\in\Omega}$ be the family of the corresponding Jacobi operators, defined as above. Let $K\geq 1$ be the constant such that $\frac{1}{K}\leq|\aw(\cdot)|\leq K$ and $0\leq|\bw(\cdot)|\leq K$. For $\omega\in\Omega$ and some $E\in\RZ\setminus \sigma(\Hw)$ set $\eta:=dist(E,\sigma(\Hw))>0$, then for each $n,m\in\GZ$ there exists a constant $\kappa:=\kappa(\eta)>0$ such that the inequality
\begin{gather*}
\left|\langle \delta_n\; \mid \;(\Hw-E)^{-1}\delta_m\rangle\right| \leq \frac{2}{\eta}\cdot e^{-\kappa|n-m|}
\end{gather*}
holds where $\delta_k(k)=1$ and $\delta_k(n)=0$ for $n\neq k$.
\end{Proposition}

\begin{proof}
Let $\kappa:=\frac{\eta}{K}\cdot c$, where $c>0$ is some constant such that $2\, c\, e^{\kappa}\leq \frac{1}{2}$ and fix an arbitrary $\omega\in\Omega$. For $k\in\GZ$ define the multiplication operator $\Mk:\ell^2(\GZ)\to \ell^2(\GZ)$ by 
\begin{gather*}
\Mk u(n) := e^{\kappa\cdot|k-n|}\cdot u(n), \quad\quad n\in\GZ.
\end{gather*}
By using for some operator $A$ on $\ell^2(\GZ)$ the equality
\begin{gather*}
\langle\delta_n \; \mid \; \left( \Mk^{-1} A \Mk \right)\delta_m\rangle = e^{-\kappa \cdot |k-n|}\cdot \langle \delta_n\;\mid\; A\; \delta_m\rangle \cdot e^{\kappa\cdot |k-m|}, \quad\quad n,m\in\GZ
\end{gather*}
it follows
\begin{gather*}
\left| \langle \delta_n\;\mid\; (\Hw-E)^{-1} \;\delta_m\rangle\right|\leq e^{-\kappa\cdot|n-m|}\cdot \left\|\left( \Mk^{-1}\Hw\Mk-E \right)^{-1}\right\|, \quad\quad n,m\in\GZ.
\end{gather*}
Applying the resolvent equation we get
\begin{gather*}
\left(\Mk^{-1}\Hw\Mk-E\right)^{-1}\cdot \left( 1+ \left(\Mk^{-1}\Hw\Mk-\Hw\right)\cdot(\Hw-E)^{-1} \right) = (\Hw-E)^{-1}.
\end{gather*}
 Recall that $|\aw(\cdot)|\leq K$ and $|\bw(\cdot)|\leq K$ for some constant $K\geq 1$ which implies that $\left|\langle\delta_n\;\mid\; \Hw\; \delta_m\rangle\right|\leq K$ for $n,m\in\GZ$. We will invert $\left( 1+ \left(\Mk^{-1}\Hw\Mk-\Hw\right)\cdot(\Hw-E)^{-1} \right)$ by using the von Neumann series. In order to do so, we have to check that the norm of \begin{gather*}
\left(\Mk^{-1}\Hw\Mk-\Hw\right)\cdot(\Hw-E)^{-1}
\end{gather*}
is smaller than one. For $n\in\GZ$ we get by a short computation that
\begin{align*}
\sum\limits_{m\in\GZ}\left| \langle\delta_n\;\mid\;\left(\Mk^{-1}\Hw\Mk-\Hw\right) \delta_m\rangle \right|&\leq \sum\limits_{\substack{m\in\GZ\\ |m-n|= 1}}\left|e^{\kappa\cdot(|k-m|-|k-n|)}-1\right| \cdot \left|\langle \delta_n\;\mid\; \Hw\;\delta_m\rangle\right|\\
&\leq \sum\limits_{\substack{m\in\GZ\\ |m-n|= 1}} \kappa\, e^{\mu} \, K\\
&\leq 2\, K\, \mu\, e^{\mu}.
\end{align*}
Consequently, 
\begin{gather*}
\left\|\Mk^{-1}\Hw\Mk-\Hw\right\|\leq 2\, K\, \kappa\, e^{\kappa}
\end{gather*}
which leads to
\begin{gather*}
\left\| \left(\Mk^{-1}\Hw\Mk-\Hw\right)\cdot(\Hw-E)^{-1} \right\|\leq 2\, K\,\kappa\, e^{\kappa}\cdot \frac{1}{\eta}  = 2\, c\, e^{\kappa}\leq \frac{1}{2}.
\end{gather*}
Hence, by the norm estimate for von Neumann series 
\begin{gather*}
\left\| \left( 1 + \left(\Mk^{-1}\Hw\Mk-\Hw\right)\cdot(\Hw-E)^{-1} \right)^{-1} \right\|\leq 2
\end{gather*}
and by the previous considerations
\begin{gather*}
\left(\Mk^{-1}\Hw\Mk-E\right)^{-1} = (\Hw-E)^{-1}\left( 1+ \left(\Mk^{-1}\Hw\Mk-\Hw\right)\cdot(\Hw-E)^{-1}\right)^{-1}.
\end{gather*}
By the definition of $\kappa$ and the fact that $\|(\Hw-E)^{-1}\|\leq\frac{1}{\eta}$ this implies
\begin{gather*}
\left| \langle \delta_n\;\mid\; (\Hw-E)^{-1} \;\delta_m\rangle\right| \leq e^{-\kappa\cdot|n-m|}\cdot \left\|\left(\Mk^{-1}\Hw\Mk-E\right)^{-1}\right\| \leq \frac{2}{\eta} \cdot e^{-\kappa\cdot|n-m|}.
\end{gather*}
\end{proof}

The next statement follows the lines of \cite{Lenz1}, Lemma 4.3 and \cite{Lenz2}, Theorem 3.

\begin{Lemma}\label{Lemma - E with gamma(E)neq 0 implies uniformity}
Let $(\Omega,T)$ be strictly ergodic and $E\in\RZ\setminus\Sigma$. Then $M^E$ is uniform and $\gamma(E)>0$.
\end{Lemma}

\begin{proof}
Let $E\in\RZ\setminus\Sigma = \RZ\setminus\sigma(\Hw)$ for $\omega\in\Omega$. In Lemma \ref{Lemma - Gamma is a subset of the spectrum} it is shown that $\Gamma\subseteq\Sigma$ is a general result and so, $\gamma(E)>0$ for $E\in\RZ\setminus\Sigma$. According to Lemma \ref{Lemma - equivalence of M und tilde(M)} we have $\Lambda(\wME)>0$.

\medskip

In the following, we will first show for $\omega\in\Omega$ that there exist two unique (up to a sign) normalized vectors $u(\omega)$ and $v(\omega)$ such that they satisfy the following condition. The norm $\|\wME(n,\omega)\; u(\omega)\|$ decays exponentially, if $n$ tends to $\infty$ and similarly, $\|\wME(-n,\omega)\; v(\omega)\|$ decays exponentially, if $n$ goes to $\infty$. Secondly, we will use these normalized vectors to construct a diagonalization of $\wME$ as in Lemma \ref{Lemma - stetig diagonalisierbar} leading to the uniformity of $M^E$.

\medskip

For $\omega\in\Omega$, set
\begin{alignat*}{4}
u_i(n):=\; & (\Hw-E)^{-1} \;\delta_i(n) \;=\; &\left\langle \delta_n\; | \; (\Hw-E)^{-1} \; \delta_i\right\rangle, \quad\quad &n\in\GZ,
\end{alignat*}
for $i\in\GZ$ which is an element of $\ell^2(\GZ)$. Fix one $\omega\in\Omega$ and consider the vectors
\begin{gather*}
\underline{u_0}\; :=\; \begin{pmatrix}
u_0(0)\\
u_0(1)
\end{pmatrix} \quad\quad\quad\quad \underline{u_{-1}}\; :=\; \begin{pmatrix}
u_{-1}(0)\\
u_{-1}(1)
\end{pmatrix}.
\end{gather*}
Note that for $u\in\ell^2(\GZ)$ and $n\in\GZ$ the value $(\Hw u)(n)$ depends only on $u(n-1), u(n)$ and $u(n+1)$. Further, $(\Hw-E)\; u_0(0)=1$ and $|u_0(-1)|=|u_{-1}(0)|$. Consequently, at least one of the numbers $u_0(0),\; u_0(1)$ and $u_{-1}(0)$ is not zero. Thus, we can normalize one of the vectors $\underline{u_0}$ and $\underline{u_{-1}}$. Without loss of generality, let $u_0$ be the vector which can be normalized. Denote its normalized vector by $u(\omega):=\frac{1}{\|\underline{u_0}\|}\cdot \underline{u_0}$. By definition we have $(\Hw-E)\; u_0(n)=0$ for $n\geq 1$ implying that $u_0$ is a solution to the right. Hence,
\begin{gather*}
\wME(n,\omega)\; u(\omega)=\frac{1}{\|\underline{u_0}\|}\cdot\begin{pmatrix}
u_0(n)\\
\aw(n+1)\cdot u_0(n+1)
\end{pmatrix}
\end{gather*}
and so,
\begin{gather*}
\|\wME(n,\omega)\; u(\omega)\| \leq \widetilde{K}\cdot \left( |u_0(n)| + |(u_0(n+1)| \right), \quad\quad n\geq 1
\end{gather*}
where $\widetilde{K}>0$ is some constant independent of $n\geq 1$. In accordance with Proposition \ref{Proposition - Combes Thomas}, for $i\in\GZ$ there exist a constant $D(i)$ and $\kappa(i)$ such that $|u_i(n)|\leq D(i)\cdot\exp(-\kappa(i)\cdot|n|)$ for $n\in\GZ$. Thus, by Proposition \ref{Proposition - Combes Thomas} we get that $\|\wME(n,\omega)\; u(\omega)\|$ decays exponentially, if $n$ goes to $\infty$. Analogously, we construct $v(\omega):=\frac{1}{\|\underline{u_i}\|}\cdot \underline{u_i}$ for $i$ equal to $1$ or $2$. As above, we get that $\|\wME(-n,\omega)\; v(\omega)\|$ decays exponentially, if $n$ tends to $\infty$.

\medskip

Since, for all $i\in\GZ$, the map $\omega\mapsto(\Hw-E)^{-1}\; \delta_i$ is continuous with respect to the topology on $\Omega$ it follows that $u(\omega)$ and $v(\omega)$ can be chosen in a continuous dependency on $\Omega$. Now we will verify that these vectors are unique up to a sign and linearily independent.

\medskip

Assume that $u(\omega)$ and $v(\omega)$ are not linearily independent meaning that there is some $0\neq c\in\RZ$ such that $u(\omega)=c\cdot v(\omega)$. Let $\left(\alpha(n)\right)_{n\in\GZ}$ be the sequence satisfying $\begin{pmatrix}
\alpha(n+1)\\
\alpha(n)
\end{pmatrix}=\wME(n,\omega) \cdot u(\omega)$ for $n\in\GZ$, which is a solution of the difference equation $(\spadesuit)$. Moreover, by the above consideration
\begin{gather*}
\|\wME(n,\omega)u(\omega)\|\leq D e^{-\kappa\cdot n}, \quad\quad n\in\NZ
\end{gather*}
and
\begin{gather*}
\|\wME(-n,\omega)u(\omega)\| = |c| \cdot \|\wME(-n,\omega)v(\omega)\| \leq |c|\cdot D e^{-\kappa\cdot n}, \quad\quad n\in\NZ
\end{gather*}
for some constants $D,\kappa>0$. Consequently, the sequence $\left(\alpha(n)\right)_{n\in\GZ}$ is an element of $\ell^2(\GZ)$. This implies that $E\in\Sigma$ contradicting the fact that $E$ was an element in the resolvent. Hence, $u(\omega)$ and $v(\omega)$ are linearily independent.

\medskip

Assume that $u(\omega)$ is not unique up to a sign. Then there are two linearily independent $u^{(1)}(\omega),u^{(2)}(\omega)\in\RZ^2$ such that $\|\wME(n,\omega)\; u^{(1)}(\omega)\|$ and $\|\wME(n,\omega)\; u^{(2)}(\omega)\|$ tend to zero. For all $x\in\RZ^2$ there are $\lambda_1,\lambda_2\in\RZ$ such that $x=\lambda_1\cdot u_1(\omega) + \lambda_2\cdot u_2(\omega)$. Thus, 
\begin{gather*}
\|\wME(n,\omega)x\|\leq |\lambda_1|\cdot\underbrace{\|\wME(n,\omega)u^{(1)}(\omega)\|}_{\rightarrow 0} + |\lambda_2|\cdot\underbrace{\|\wME(n,\omega)u^{(2)}(\omega)\|}_{\rightarrow 0} \rightarrow 0,\; n\rightarrow\infty
\end{gather*}
which contradicts the fact that $\|\wME(n,\omega)\|\geq 1$ for all $n\in\NZ$. This implies that $u(\omega)$ is unique up to a sign and similarly $v(\omega)$ is unique up to a sign. Denote by $U(\omega)$ the corresponding unique one-dimensional subspace of $\RZ^2$ generated by $u(\omega)$ and analogously $V(\omega)$ for $v(\omega)$. 

\medskip

Next, define a matrix $C(\omega):=\left( u(\omega),v(\omega)\right)$. According to the previous considerations, the matrix $C(\omega)$ is invertible. As mentioned in the beginning of Section \ref{section - Generalities}, we know that $\wME(n,T\omega)\wME(\omega)=\wME(n+1,\omega)$. Thus, $\|\wME(n,T\omega)x(T\omega)\|$ is exponentially decaying for the vector $x(T\omega):=\wME(\omega)u(\omega)$. As we have seen above, there can be at most one one-dimensional subspace $U(T\omega) \subsetneq \RZ^2$ such that the solutions decay exponentially for $T\omega\in\Omega$. Consequently, $x(T\omega)$ is an element of $U(T\omega)$ and so, there exists a $d(\omega)\in\RZ$ such that $\wME(\omega)u(\omega)= d(\omega)\cdot u(T\omega)$, where $u(T\omega)$ is the unique vector (up to a sign) with norm one for $T\omega\in\Omega$. Analogously, there exists an $e(\omega)\in\RZ$ such that $\wME(\omega)v(\omega) = e(\omega)\cdot v(T\omega)$. Hence,
\begin{align*}
C^{-1}(T\omega)\wME(\omega)C(\omega) = &\left(C^{-1}(T\omega)\wME(\omega)u(\omega)\; , \; C^{-1}(T\omega)\wME(\omega)v(\omega)\right)\\
= &\left(d(\omega)\cdot C^{-1}(T\omega)u(T\omega))\; , \; e(\omega)\cdot C^{-1}(T\omega) v(T\omega)\right)\\
= &\left( \begin{matrix}
d(\omega) & 0\\
0 & e(\omega)
\end{matrix}\right).
\end{align*}

Multiplying $u(\omega)$, $v(\omega)$ or both with minus one will add a minus sign to $d(\omega)$ respectively $e(\omega)$. Since, further, $u(\omega)$ and $v(\omega)$ can be chosen continuously in a neighborhood of $\omega\in\Omega$, by changing the sign, it follows that the maps $\|C\|, \|C^{-1}\|, |d|, |e|:\Omega\to\RZ$ are continuous. The map $\wME(\omega)$ is invertible impying that $d(\omega)$ and $e(\omega)$ never vanish. Thus, the function $\wME$ is uniform by Lemma \ref{Lemma - stetig diagonalisierbar}. According to Lemma \ref{Lemma - equivalence vanishing of Lambda and uniformity}, the continuous function $M^E$ is uniform as well.
\end{proof}

\section{The main results}\label{section - The main results}

In our main Theorem \ref{Theorem - the spectrum is equal to the union}, we give a complete description of the spectrum of Jacobi operators in terms of the Lyapunov exponent $\gamma(E)$ and the uniformity properties of the transfer matrices $M^E$, $(E \in \mathbb{R})$. As a special case, we show in Corollary \ref{Corollary - equivalence uniformity} that uniformity of all the $M^E$ is equivalent to the fact that the spectrum is exactly the set of zeros of the Lyapunov exponent $\gamma(\cdot)$ as a function of $E \in \mathbb{R}$. Furthermore, the uniformity condition guarantees the continuity of the function $\gamma(\cdot)$.

\begin{Theorem}\label{Theorem - the spectrum is equal to the union}
Let $(\Omega,T)$ be strictly ergodic and consider a family of Jacobi operators $(\Hw)_{\omega\in\Omega}$. Then the spectrum $\Sigma$ is equal to the disjoint union
\begin{gather*}
\{E\in\RZ\;|\; \gamma(E)=0\}\bigsqcup\{E\in\RZ\;|\; M^E \text{ is not uniform}\}.
\end{gather*}
\end{Theorem}

\begin{proof}
It follows from Lemma \ref{Lemma - E with gamma(E)=0 implies uniformity} that these sets are disjoint. It suffices to show that the equation
\begin{gather*}
\Sigma^C = \{E\in\RZ\; | \; M^E \text{ is uniform and } \gamma(E)>0\}
\end{gather*}
holds. This follows immediately by Proposition \ref{Proposition - E nicht im Spektrum} and Lemma \ref{Lemma - E with gamma(E)neq 0 implies uniformity}. 
\end{proof}

\begin{Corollary}\label{Corollary - equivalence uniformity}
Let $(\Omega,T)$ be strictly ergodic. Then the following assertions are equivalent.
\begin{itemize}
\item[(i)] The matrix $M^E$ is uniform for each $E\in\RZ$.
\item[(ii)] $\Sigma = \{E\in\RZ \; |\; \gamma(E)=0\}$
\end{itemize}
In this case, $\gamma:\RZ\rightarrow [0,\infty)$ is continuous.
\end{Corollary}

\begin{proof}
This equivalence follows immediately by Theorem \ref{Theorem - the spectrum is equal to the union}. The continuity of $\gamma$ is a consequence of Lemma \ref{Lemma - M uniform impliziert}.
\end{proof}

\section{Subshifts}\label{section - Subshifts}

We will now apply our previous results to the special case of a dynamical system induced by a subshift. To do so, we first recapitulate some well-known results about the spectrum of the family of Jacobi operators corresponding to a subshift. Then our main statement will be that the spectrum of a Jacobi operator, generated by an aperiodic subshift, with some reasonable requirements is a Cantor set of Lebesgue measure zero, if the transfer matrix $M^E$ is uniform for each $E\in\RZ$ (Theorem \ref{Theorem - aperiodic subshifts has spectrum of Lebesgue measure zero}). The main idea of the proof is to apply an assertion of \cite{Remling}. Further, we recall the notion of the Boshernitzan condition for subshifts. Indeed, a large class of subshifts satisfies this condition. It turns out that in this case the transfer matrices are uniform for all energies.

\medskip

Consider a finite alphabet $\A\subsetneq\RZ$ and $\A^{\GZ}:=\{\varphi:\GZ\rightarrow\A\}$. Denote by $d_{\A}:\A\times\A\rightarrow\{0,1\}$ the discrete metric on $\A$. We define a metric $d:\A^{\GZ}\times\A^{\GZ}\rightarrow[0,\infty)$ on $\A^{\GZ}$ by
\begin{gather*}
d(\varphi,\psi):=\sum\limits_{k=-\infty}^{\infty}\frac{d_{\A}(\varphi(k),\psi(k))}{2^{|k|}}.
\end{gather*}
Then a well-known result is that $(\A^{\GZ},d)$ is compact, see \cite{Walthers}.

\medskip

Let $(\Omega,T)$ be a subshift over $\A$, where $\Omega$ is a closed subset of $\A^{\GZ}$ and invariant under the homeomorphism $T:\A^{\GZ}\rightarrow\A^{\GZ}$ defined by
\begin{gather*}
(T\omega)(n):= \omega(n+1), \quad\omega\in\Omega.
\end{gather*}
This map is also called the \textit{shift operator}. For each $\omega\in\Omega$ we have the set of words associated to $\omega$ given by 
\begin{gather*}
\W_{\omega}:=\{ \omega(l)\ldots\omega(l+n-1) \;|\; l\in\GZ,n\in\NZ \}.
\end{gather*}
Further, $\W(\Omega):=\underset{\omega\in\Omega}{\bigcup}\W_{\omega}$ is the set of words associated to $\Omega$. We say a subshift $(\Omega,T)$ is \textit{aperiodic}, if for all $\omega\in\Omega$ there is no $0\neq m\in\NZ$ such that $T^n \omega=\omega$. If for $\omega\in\Omega$ there exists a $0\neq m\in\NZ$ such that $T^m\omega=\omega$ this element is called \textit{periodic} and $m$ is the \textit{period} of $\omega$.

\medskip

First of all we recapitulate some well-known results, see e.g. the textbooks \cite{Teschl,CarmonaLacroix}.

\begin{Proposition}\label{Proposition - Constancy of the spectrum a.s.}
Let $(\Omega,T)$ be a uniquely ergodic dynamical system induced by a subshift. Then there are $\Sigma_{ac},\Sigma_{sc},\Sigma_{pp}\subset\RZ$ such that
\begin{align*}
\Sigma_{ac} = &\sigma_{ac}(\Hw)\; \text{ a.s.},\\
\Sigma_{sc} = &\sigma_{sc}(\Hw)\; \text{ a.s.},\\
\Sigma_{pp} = &\sigma_{pp}(\Hw)\; \text{ a.s.}.
\end{align*}
Further, for $\mu$-almost every $\omega\in\Omega$ the set $\sigma(\Hw)$ has no discrete points.
\end{Proposition}

Now the following assertion immediately follows by Proposition \ref{Proposition - Constancy of the spectrum}.

\begin{Corollary}\label{Corollary - Absence of discrete points in the spectrum}
Let $(\Omega,T)$ be a strictly ergodic dynamical system induced by a subshift. Then for every $\omega\in\Omega$ the set $\sigma(\Hw)$ does not contain a discrete point.
\end{Corollary}

\begin{Remark} If $(\Omega,T)$ is minimal, a stronger result can be shown, namely that $\Sigma_{ac}=\sigma_{ac}(\Hw)$ for all $\omega\in\Omega$, see \cite{LastSimon}, Theorem 6.1.
\end{Remark}

Consider for a subshift $(\Omega,T)$ the restriction $(\Omega^+,T^+)$ respectively $(\Omega^-,T^-)$ defined as follows:
\begin{align*}
\Omega^+:= &\{\omega\mid_{\NZ_0} \;|\; \omega\in\Omega\} \text{ with } T^+:=  T,\\
\Omega^-:= &\{\omega\mid_{\GZ\setminus\NZ_0} \;|\; \omega\in\Omega\} \text{ with } T^-:=  T^{-1}.
\end{align*}

Recall for $\omega\in\Omega$ the definition of the Jacobi operator $\Hw$ with the corresponding continuous maps $p:\Omega\to\RZ\setminus\{0\}$ and $q:\Omega\to\RZ$. For $\omega^+\in\Omega^+$ denote by $\Hw^+$ the restriction of the Jacobi operator $\Hw$ where $\omega^+=\omega\mid_{\NZ_0}$. Similarly, denote the corresponding restrictions of $p$ and $q$ by $p^+$ and $q^+$, which are still continuous.

\medskip

A sequence $u:=(u(n))_{n\in\NZ_0}$ is called \textit{eventually periodic} if there exist $K,m\in\NZ$ such that $u$ is periodic outside of $\{0,1,\ldots,K\}$ with period $m$. For an $\omega^+\in\Omega^+$ eventual periodicity is defined accordingly by considering the sequence $((T^+)^n\omega^+)_{n\in\NZ_0}$. A closed set is called a \textit{Cantor set}, if it does not contain a non-trivial interval and no discrete points. Note that a set of Lebesgue measure zero cannot have a non-trivial interval at all.

\begin{Lemma}\label{Lemma - aperiodic transfers to restriction of TDS}
Let $(\Omega,T)$ be an aperiodic subshift. Then $(\Omega^+,T^+)$ and $(\Omega^-,T^-)$ do not contain an eventually periodic element.
\end{Lemma}

\begin{proof}
We show that $(\Omega^+,T^+)$ does not have an eventually periodic element. The proof for $(\Omega^-,T^-)$ works similarly. Assume that there exists an eventually periodic element $\omega^+$ with period $m\in\NZ$. Then there is an $\omega\in\Omega$ such that $\omega^+=\omega\mid_{\NZ_0}$. Consider the sequence $\omega_n:=T^{-n\cdot m}\omega$. Then by compactness of $\Omega$ there is a convergent subsequence $\omega_{n_j}$ converging to some $z\in\Omega$. Since $\omega^+$ is eventually peridic, it follows that $z$ is periodic, contradicting the fact that $(\Omega,T)$ is aperiodic.
\end{proof}

Consider the dynamical system $(\widetilde{\Omega},\widetilde{T})$ defined by
\begin{gather*}
\widetilde{\Omega}:=\left\{ \widetilde{\omega}:=\begin{pmatrix}
p(\omega)\\
q(\omega)
\end{pmatrix}\;|\; \omega\in\Omega \right\}
\end{gather*}
endowed with the product topology and 
\begin{gather*}
\widetilde{T}\widetilde{\omega} := \begin{pmatrix}
p(T\omega)\\
q(T\omega)
\end{pmatrix}, \quad\quad \widetilde{\omega}\in\widetilde{\Omega}.
\end{gather*}

Note that the notion of periodicity and eventual periodicity carries over for $\omega\in\Omega$ to $p(\omega)$ and $q(\omega)$ in the obvious way. Precisely, $p(\omega)$ is periodic, if there exists a $0\neq m\in\NZ$ such that $p(T^m\omega)=p(\omega)$.

\medskip

If $(\Omega,T)$ is minimal, the dynamical system $(\widetilde{\Omega},\widetilde{T})$ is minimal as well. 
Our aim is to show that the spectrum of a family of Jacobi operators is supported on a Cantor set of Lebesgue measure zero, if $p(\omega)$ or $q(\omega)$ is not periodic for each $\omega\in\Omega$. In this section, we consider Jacobi operators $\Hw$ associated with an aperiodic subshift such that the aperiodicity of the subshifts carries over to $(\widetilde{\Omega},\widetilde{T})$. Denote this condition by (A). This is for example the case, if $p$ or $q$ is injective.

\begin{Theorem}\label{Theorem - aperiodic subshifts has spectrum of Lebesgue measure zero}
Let $(\Omega,T)$ be a strictly ergodic and aperiodic subshift. Consider the continuous maps $p:\Omega\to\RZ\setminus\{0\}$ and $q:\Omega\to\RZ$ which take finitely many values with corresponding family of Jacobi operators $(\Hw)_{\omega\in\Omega}$. Suppose that condition (A) is satisfied and that the transfer matrix $M^E$ is uniform for every $E\in\RZ$. Then the spectrum $\Sigma$ is a Cantor set of Lebesgue measure zero.
\end{Theorem}

\begin{proof}
Since $\Sigma$ does not contain discrete points, we have to verify that $|\Sigma|=0$ where $|\cdot|$ denotes the Lebesgue measure, see Corollary \ref{Corollary - Absence of discrete points in the spectrum}. By Corollary \ref{Corollary - equivalence uniformity} it is enough to verify that $|\Gamma|=0$. According to general results, (see e.g. \cite{Teschl}, Theorem 5.17) it is sufficient to show that $\Sigma_{ac}$ is empty.

\medskip

Let $\omega\in\Omega$ be chosen such that $\sigma_{ac}(\Hw)=\Sigma_{ac}$. Assume that $\sigma_{ac}(\Hw)$ is non-empty. Since any perturbation of finite range does not change the absolutely continuous spectrum it follows that $\sigma_{ac}(\Hw^+)$ or $\sigma_{ac}(\Hw^-)$ is non-empty. Without loss of generality, let $\sigma_{ac}(\Hw^+)\neq\emptyset$ and so, $p^+(T^{(\cdot)}\omega)$ and $q^+(T^{(\cdot)}\omega)$ are eventually periodic, see \cite{Remling}, Theorem 1.1. Hence, there is an eventually periodic element of the dynamical system $(\widetilde{\Omega}^+,\widetilde{T}^+)$. However, this contradicts the assertion of Lemma \ref{Lemma - aperiodic transfers to restriction of TDS}.
\end{proof}

In the work \cite{DamanikLenz} it is shown that the so called Boshernitzan condition, first introduced in \cite{Boshernitzan1}, implies for a minimal subshift that a locally constant (definition see below), continuous function $A:\Omega\to SL(2,\RZ)$ is uniform. Indeed, a large class of subshifts satisfies this condition. For instance, in the work \cite{DamanikLenz} it is shown that subshifts obeying positive weights (PW) satisfies the Boshernitzan condition. 
The class of linear repetitive (or linear recurrent) subshifts is contained in the class of subshifts with (PW). Actually, it turns out that these two classes are equal by unpublished results of Boshernitzan, see \cite{BesbesBoshernitzanLenz} as well. Also all Sturmian subshifts fulfill the Boshernitzan condition. 
Further, almost all interval exchange transformations, 
almost all circle maps 
and almost all Arnoux-Rauzy subshifts 
satisfy the Boshernitzan condition, see \cite{DamanikLenz2}.

\medskip

A continuous function $M$ on $\Omega$ is called \textit{locally constant}, if there exists an $N\in\NZ$ such that for $\omega_1,\omega_2\in\Omega$ whenever $(\omega_1(-N),\ldots,\omega_1(N))= (\omega_2(-N),\ldots,\omega_2(N))$, then the equality $M(\omega_1)=M(\omega_2)$ holds.

\medskip

Let $(\Omega,T)$ be a subshift over the finite alphabet $\A\subsetneq \RZ$ and $\W(\Omega)$ be the set of the associated words of $\Omega$. For $v\in\Omega$ we define the subset of all elements of $\Omega$ which begin with the associated word $v$ by
\begin{gather*}
\mathcal{V}_{v}:=\{\omega\in\Omega\;|\; \omega(1)\cdots\omega(|v|)=v\}.
\end{gather*}
Note that $|v|$ denotes the length of the associated word $v\in\W(\Omega)$. Further, for a $T$-invariant probability measure $\mu$ on $\Omega$ and $n\in\NZ$ we define the number
\begin{gather*}
\eta_{\mu}(n):=\min\{ \mu(\mathcal{V}_{v})\;|\; v\in\W(\Omega), |v|=n \}.
\end{gather*}
The following definition was originally introduced by Boshernitzan in his work \cite{Boshernitzan1}. A subshift $(\Omega,T)$ over the finite alphabet $\A\subsetneq\RZ$ satisfies the \textit{Boshernitzan condition}, if there exists an ergodic probability measure $\mu$ on $\Omega$ such that
\begin{gather*}
\limsup_{n\to\infty} n\cdot \eta_{\mu}(n)>0.
\end{gather*}
As mentioned before, a large class of subshifts fulfills this condition and the following statement, proven in \cite{DamanikLenz}, gives us a useful tool.

\begin{Theorem}\label{Theorem - (B) implies uniformity}
Let $(\Omega,T)$ be a minimal subshift over the finite alphabet $\A\subsetneq\RZ$ satisfying the Boshernitzan condition. Then a locally constant map $M:\Omega\to SL(2,\RZ)$ is uniform.
\end{Theorem}

{Theorem \ref{Theorem - (B) implies uniformity} provides a sufficient condition for the uniformity of all transfer matrices corresponding to a Jacobi operator and its subshift. Combining this with Theorem \ref{Theorem - aperiodic subshifts has spectrum of Lebesgue measure zero} we get the following assertion.

\begin{Corollary}
Let $(\Omega,T)$ be a minimal, aperiodic subshift such that the Boshernitzan condition holds. Consider the family of the corresponding Jacobi operators $\{\Hw\}_{\omega\in\Omega}$ where the continuous maps $p$ and $q$ take finitely many values and they obey condition (A). Then the transfer matrix $M^E:\Omega\to GL(2,\RZ)$ is uniform for each $E\in\RZ$. In particular, the spectrum $\Sigma$ is a Cantor set of Lebesgue measure zero. 
\end{Corollary}

\begin{proof}
It is shown in \cite{Boshernitzan2} (Theorem 1.2) that the Boshernitzan condition for a minimal subshift implies unique ergodicity of the subshift. Since $p$ and $q$ are uniformly continuous and only take finitely many values, it follows that they are locally constant. Thus, for $E\in\RZ$ the continuous map $\wME:\Omega\to SL(2,\RZ)$ is locally constant as well. According to Theorem \ref{Theorem - (B) implies uniformity} it follows that $\wME$ is uniform for each $E\in\RZ$ and so is $M^E$ as well, see Lemma \ref{Lemma - equivalence of M und tilde(M)}. Consequently, we can apply Theorem \ref{Theorem - aperiodic subshifts has spectrum of Lebesgue measure zero} leading to our assertion.
\end{proof}

\section*{Acknowledgment}
We would like to express our thanks to our advisor {\sc Daniel Lenz} for his careful guidance and support. Especially, we thank him for coming up with the issue of this work and for fruitful suggestions concerning the preparation of the paper. It is our great pleasure to acknowledge {\sc Matthias Keller} for enlightening discussions about the connection of spectral aspects and subexponentially growing functions. FP would like to thank the German National Academic Foundation (Studienstiftung des deutschen Volkes) for financial support. 

\bibliographystyle{amsalpha}
\bibliography{PS_lit}
\end{document}